\documentclass[pra,aps,reprint,a4paper,superscriptaddress,floatfix,longbibliography]{revtex4-2}
\usepackage{times}
\usepackage{graphicx}
\usepackage{epsfig}
\usepackage{amsfonts}
\usepackage{amsmath}
\usepackage{amssymb,bbm}
\usepackage{dsfont}
\usepackage{amsthm}% in order to use a black box at the end of proofs
\usepackage{color}
\usepackage[colorlinks=true,linkcolor=blue,citecolor=blue,urlcolor=blue]{hyperref}
\usepackage{comment}
\usepackage{braket}
\usepackage{physics}
\usepackage{multirow}
\usepackage{diagbox}
\newtheorem{theorem}{Theorem}

\newtheorem{proposition}{Proposition}

\newcommand{\id}{\mathds{1}}

\begin{document}

%%%%%%%%%%%%%%%%%%%%%%%%%%%%%%%%%%%%%%%%%%%%%%%%%%%%%%%%%%%%%%%%%%%

\title{Supersinglets can be self-tested with perfect quantum strategies}

%%%%%%%%%%%%%%%%%%%%%%%%%%%%%%%%%%%%%%%%%%%%%%%%%%%%%%%%%%%%%%%%%%%

\author{Debashis Saha}
\email{saha@iisertvm.ac.in}
\affiliation{School of Physics, Indian Institute of Science Education and Research Thiruvananthapuram, Kerala 695551, India}
\affiliation{Department of Physics, School of Basic Sciences, Indian Institute of Technology Bhubaneswar,
Bhubaneswar, Odisha 752050, India}

\author{Ad\'an~Cabello}
\email{adan@us.es}
\affiliation{Departamento de F\'{\i}sica Aplicada II, Universidad de Sevilla, E-41012 Sevilla,
Spain}
\affiliation{Instituto Carlos~I de F\'{\i}sica Te\'orica y Computacional, Universidad de
Sevilla, E-41012 Sevilla, Spain}

%%%%%%%%%%%%%%%%%%%%%%%%%%%%%%%%%%%%%%%%%%%%%%%%%%%%%%%%%%%%%%%%%%%

\begin{abstract}
Supersinglets are states of spin-zero of $d \ge 3$ particles of $d$ levels. They are invariant under unitary transformations of the form $U^{\otimes d}$ and have applications in metrology, error protection, and communication.
They also violate some specific Bell inequalities. However, none any of these applications {\em require} supersinglets nor do any of these Bell inequality violations capture the unique properties of the supersinglets. This leads to two questions. Question 1 is whether there exists a task that can be solved only with supersinglets. Question 2 is whether supersinglets can produce a unique $d$-partite, $d$-dimensional nonlocal signature. We answer both questions affirmatively by presenting a protocol that self-test all supersinglets by producing $d$-partite, $d$-dimensional {\em perfect} quantum strategies for any $d \ge 3$.
\end{abstract}

%%%%%%%%%%%%%%%%%%%%%%%%%%%%%%%%%%%%%%%%%%%%%%%%%%%%%%%%%%%%%%%%%%%

\maketitle

%%%%%%%%%%%%%%%%%%%%%%%%%%%%%%%%%%%%%%%%%%%%%%%%%%%%%%%%%%%%%%%%%%%

\section{Introduction}

Supersinglets \cite{Cabello:2002PRL,Cabello:2003JMP,Jin:2005PRA,Quiang:2011PLA,Chen:2016SR,PhysRevA.106.033314} are states of total spin
zero of $d \ge 3$ particles of $d$ levels. They can be written as
\begin{equation}
|{\cal S}_d^{(d)}\rangle=\frac{1}{\sqrt{d!}}
\sum_{\scriptscriptstyle{ {\stackrel{\scriptscriptstyle{\rm
permutations}} {{\rm of}\;(0,1,\ldots d-1)}}}} \!\!\!\!\!\!
\varepsilon_{a_0 a_1 \ldots a_{d-1}} \left| a_0 a_1 \ldots a_{d-1} \right\rangle,
\label{SNSN0}
\end{equation}
where $\varepsilon_{a_0 a_1 \ldots a_{d-1}}$ is the Levi-Civita symbol, which is $+1$ or $-1$ depending on whether $(a_0, a_1, \ldots, a_{d-1})$ is an even or odd permutation of $(0,1,\ldots,d-1)$.
The name $d$-qudit-supersinglets follows from that they generalize, to more particles and higher dimensions, the two-qubit singlet state, $|{\cal S}_2^{(2)}\rangle= |\psi^-\rangle = \frac{1}{\sqrt{2}} (|01\rangle - |10\rangle)$, ubiquitous in quantum information. Physically, $|{\cal S}_d^{(d)}\rangle$ are the states when a spin zero particle decays into $d$ particles of spin $(d-1)/2$. Remarkably, $|{\cal S}_d^{(d)}\rangle$ are invariant under the tensor product of $d$ equal unitary operations, that is,
\begin{equation} \label{prop1}
U^{\bigotimes d} |{\cal S}_d^{(d)}\rangle= |{\cal
S}_d^{(d)}\rangle,
\end{equation}
where $U$ is a single-particle unitary operation. This property makes supersinglets useful for protecting quantum information in decoherence-free subspaces \cite{Bourennane:2004PRL,Cabello:2007PRA}, metrology \cite{PhysRevA.106.033314}, producing eigenstates of unknown unitary operators \cite{Hillery:2001PRA}, and communications tasks, including 
Byzantine agreement \cite{Fitzi:2001PRL}, secret sharing \cite{Cabello:2002PRL}, the $n$-strangers problem \cite{Cabello:2002PRL}, and the ``liar detection'' \cite{Cabello:2002PRL}. 

However, there are several questions about supersinglets for which we still have no answers. One of them is what can we do with supersinglets that is not possible with any other quantum state. This question is especially pertinent when we realize that none of the applications mentioned {\em require} supersinglets: each of them can be accomplished with simpler quantum states. This leads to {\em Question 1}:
is there a task that can only be accomplished with supersinglets? This question is formally equivalent to identifying a protocol that {\em self-tests} \cite{Yao_self,Supic:2020Q} supersinglets, that is, that produces a correlation that is a unique (up to local isometries) signature of the supersinglets. Self-testing protocols exist for all bipartite pure states \cite{Coladangelo:2017NC} and for all pure multipartite entangled states of qubits \cite{balanz:2024XXX}, but not for multipartite entangled states of high-dimensional particles. 

The second question is why supersinglets are special besides the features mentioned before. Specifically, why are they special in terms of nonlocality and entanglement. While it is known that supersinglets violate some Bell inequalities \cite{Cabello:2002PRL,Grandjean:2012PRA,Laskowski2014} and thus provide quantum advantage in some multipartite nonlocal games, an open question, {\em Question 2}, is whether supersinglets allow for $d$-partite, $d$-dimensional {\em perfect quantum strategies} or {\em pseudo-telepathy} \cite{GBT05}, that is, whether supersinglets allow $d$ parties, which cannot communicate to each other, to win {\em every} round of a $d$-partite, $d$-dimensional nonlocal game, as occurs with Greenberger-Horne-Zeilinger \cite{GHZ89,Mermin:1990AJP} and related states \cite{Guhne:2005PRL,Cabello:2008PRA} for $n$-partite two-dimensional games. The general question of which states allow for perfect quantum strategies is by itself an open problem \cite{Mancinska}. The importance of perfect quantum strategies goes far beyond nonlocal games. On the one hand, they are key tools for proving results such as the quantum computational advantage for shallow circuits \cite{Bravyi:2018SCI}, the solution to Tsirelson's problem \cite{Ji:2021CACM}, and the impossibility of classically simulating quantum correlations with arbitrary relaxations of measurement and parameter independence \cite{Vieira;2024XXX}. On the other hand, it has been recently proven \cite{Liu:2024PRR} that the existence of a bipartite perfect quantum correlation is equivalent to the existence of a quantum correlation with maximal nonlocal content \cite{Elitzur:1992PLA} or fully nonlocal correlation \cite{Aolita:2012PRA}, of a Greenberger-Horne-Zeilinger-like proof of Bell's theorem \cite{GHZ89,Mermin:1990AJP,Cabello:2001PRLb}, of a quantum correlation in a face of the nonsignaling polytope with no local points \cite{Liu:2024PRR}, and of a special type of Kochen-Specker set \cite{Cabello:2025PRL}.

Regarding entanglement, the $d$-$d$-supersinglets have, at the same time, genuinely high-dimensional entanglement (i.e., they cannot be generated by entangling subsystems of dimension $d' < d$) and genuinely multipartite entanglement (i.e., they cannot be generated by entangling only $n < d$ of the particles) \cite{Cobucci:2024SA}. {\em Question 3} is: are supersinglets {\em the} maximally genuinely $d$-partite and genuinely $d$-dimensional entangled states? \cite{Cobucci:2024SA}. An affirmative answer to any of these questions would push the experimental interest on supersinglets beyond the current theoretical stage \cite{Jin:2005PRA,Quiang:2011PLA,Chen:2016SR,PhysRevA.106.033314}.

The aim of this paper is to answer affirmatively questions 1 and 2. Moreover, we show that there is a single approach that answers both questions simultaneously, as, for any $d$-$d$-supersinglet with $d \ge 3$, there is a self-testing protocol in which the signature is a $d$-partite, $d$-dimensional perfect quantum strategy.

%%%%%%%%%%%%%%%%%%%%%%%%%%%%%%%%%%%%%%%%%%%%%%%%%%%%%%%%%%%%%%%%%%%

\section{Perfect strategies with Supersinglets}

%%%%%%%%%%%%%%%%%%%%%%%%%%%%%%%%%%%%%%%%%%%%%%%%%%%%%%%%%%%%%%%%%%%

Here, we show how to produce $d$-partite $d$-dimensional perfect quantum strategies using $d$-$d$ supersinglets. 
Our method has two ingredients: the symmetry properties of the $d$-$d$ supersinglets and Kochen-Specker (KS) sets \cite{Kochen:1967JMM} in a Hilbert space ${\cal H}=\mathbbm{C}^d$, with $d\ge 3$.

A {\em KS set} \cite{Kochen:1967JMM} is a finite set of rank-one projectors (observables) in a Hilbert space ${\cal H}=\mathbbm{C}^d$, with finite $d \ge 3$, which does not admit an assignment of $0$ or $1$ satisfying (i) two orthogonal projectors cannot both be assigned $1$, and (ii) for every set of $d$ mutually orthogonal projectors, one of them must be assigned $1$.

We will describe our method by using the KS set in ${\cal H}=\mathbbm{C}^4$ shown in Fig.~\ref{fig1}. The reason for this choice is that this is the KS set with the smallest number of vectors in any dimension \cite{Cabello:1996PLA,Xu:2020PRL}. The method works equally by using any complete KS set in $\mathbbm{C}^d$, where $d$ is the number of parties (and the number of levels of their quantum systems).

A KS set is {\em complete} \cite{Xu:2024PRL} if every pair of orthogonal projectors is
in a set of $d$ mutually orthogonal projectors. Compact and symmetric KS sets are known in $d=3$ \cite{Peres:1991JPA}, $d=4$ \cite{Peres:1991JPA,Cabello:1996PLA}, $d=6$ \cite{LisonekPRA2014}, $d=8$ \cite{Kernaghan:1995PLA}, and $d = 2^kp^m$ for $p$ prime, $k \in \{1,2\}, m \geq 1$ (as well as $d = 8p$ for $p \geq 19$, and other sporadic examples) \cite{Trandafir:2024XXX}. There are also methods to produce KS in any finite $d \ge 3$ \cite{Cabello:1996JPA,Cabello:2005PLA}. Not all these KS sets are complete. However, completing each of them is straightforward \cite{Xu:2024PRL}.

%%%%%%%%%%%%%%%%%%%%%%%%%%%%%%%%%%%%%%%%%%%%%%%%%%%%%%%%%%%%%%%%%%%
% Fig. 1
%%%%%%%%%%%%%%%%%%%%%%%%%%%%%%%%%%%%%%%%%%%%%%%%%%%%%%%%%%%%%%%%%%%

\begin{figure}[t!]
 \centering
 \includegraphics[scale = 0.54]{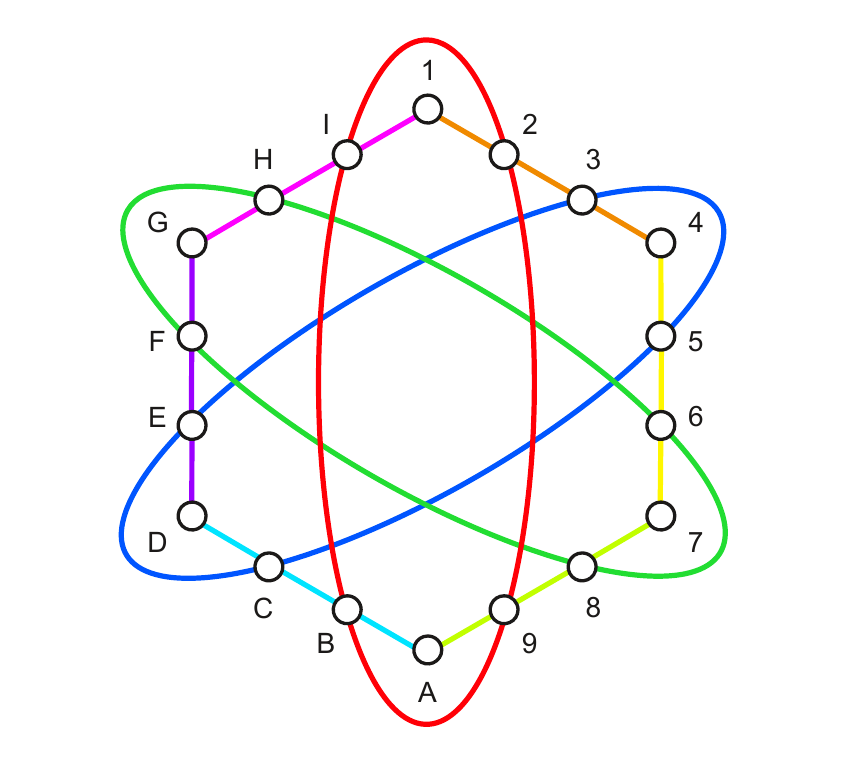}
 \caption{Relations of orthogonality between the elements of the KS set with the smallest number of rank-one observables (or vectors): the 18-vector nine-basis set \cite{Cabello:1996PLA}. Four dots in a line of the same color represent a tetrad of mutually orthogonal four-dimensional vectors. Each vector is in two tetrads and the total number of tetrads is odd. Therefore, it is impossible any assignment satisfying that, for
every set of four mutually orthogonal projectors, only one of them
must be assigned $1$. A quantum realization of the set is the following: $1=(1,0,0,0)$, $2=(0,1,0,0)$, $3=(0,0,1,1)$, $4=(0,0,1,-1)$, $5=(1,-1,0,0)$, $6=(1,1,-1,-1)$, $7=(1,1,1,1)$, $8=(1,-1,1,-1)$, $9=(1,0,-1,0)$, $A=(0,1,0,-1)$,
 $B=(1,0,1,0)$, $C=(1,1,-1,1)$, $D=(-1,1,1,1)$, $E=(1,1,1,-1)$, $F=(1,0,0,1)$, $G=(0,1,-1,0)$, $H=(0,1,1,0)$, $I=(0,0,0,1)$.}
 \label{fig1}
\end{figure}

%%%%%%%%%%%%%%%%%%%%%%%%%%%%%%%%%%%%%%%%%%%%%%%%%%%%%%%%%%%%%%%%%%%

For $d=4$, the game is as follows. Three of the four players, Alice, Bob, and Charlie, receive as inputs the same randomly chosen tetrad of orthogonal vectors of the KS set in Fig.~\ref{fig1}. The fourth player, David, receives as input a single vector randomly chosen from that tetrad. Each of Alice, Bob, and Charlie outputs two bits indicating which of their four vectors is assigned value $1$ (implicitly, the other three vectors are assigned
$0$). David outputs a bit assigning $1$ or $0$ to his vector. 
The winning condition is that Alice, Bob, and Charlie assign $1$ to three different vectors of the tetrad and David assigns $0$ if he has received one of these three or $1$ if he has received the fourth vector of the tetrad.

There is no perfect classical strategy for this game. A perfect classical strategy would imply that David can assign $1$ and $0$ to the vectors of the KS set satisfying (i) and (ii), something that is impossible by definition of KS set.
However, the following strategy is a perfect quantum strategy. The four players share a $\left|{\cal S}_4^{(4)}\right\rangle$. Each of Alice, Bob, and Charlie measures on its particle the projectors onto the vectors of the triad they receive. On his particle, David measures the rank-one projector on the vector he received.

This implies that, in any of the nine bases (tetrads) in Fig.~\ref{fig1}, $|{\cal S}_4^{(4)}\rangle$ has 
the same expression.
%as in Eq.~\eqref{s4s4}. 
For example, in the $\{D,E,F,G\}$ basis in Fig.~\ref{fig1}, 
\begin{eqnarray}
|{\cal S}_4^{(4)} \rangle & = & {1 \over 2 \sqrt {6}} (
| DEFG \rangle - | DEGF \rangle - |
DFEG \rangle + | DFGE \rangle \nonumber \\ & & +
| DGEF \rangle - | DGFE \rangle
- | EDFG \rangle + | EDGF \rangle \nonumber
\\ & & + | EFDG \rangle - | EFGD \rangle -
| EGDF \rangle + | EGFD \rangle
\nonumber \\ & & + | FDEG \rangle - | FDGE
\rangle -
| FEDG \rangle + | FEGD \rangle \nonumber \\
& & + | FGDE \rangle - | FGED \rangle
- |
GDEF \rangle + | GDFE \rangle \nonumber \\
& & + | GEDF \rangle - | GEFD \rangle -
| GFDE \rangle + | GFED \rangle).
\label{s4s4b}
\end{eqnarray}
Therefore, every time Alice, Bob, and Charlie measure the same basis and David measures one element of that basis, the winning condition of the game is satisfied.

The above game can be generalized for any complete KS set in $\mathbbm{C}^d$ as follows.

\begin{proposition}\label{prop1}
For any orthogonality graph with a realization of a complete KS set in $\mathbbm{C}^d$, there exists a $d$-party game that can be perfectly accomplished using $d$-$d$ supersinglet.
\end{proposition}
\begin{proof}
Consider an orthogonality graph that realizes a complete KS set of vectors. Let the graph have $n$ vertices and $m$ contexts (or $d$-cliques). The vectors corresponding to each context form a basis. For any natural number $k$, here we use the notation of the set $[k]:= \{0,1,\ldots,k-1\}$. Let us denote the contexts as $C_x$, where $x \in [m],$ and each context $C_x$ comprising $d$ elements from $[n]$ such that the corresponding vectors form a basis. For instance, if $C_x = \{s_1,\ldots,s_d\}$, where $s_i\in [n]$, then $\{\ket{\psi_{s_i}}\}_{i=1}^d$ forms a basis. 

In the game, a context is randomly drawn from and given as input to the first $(d-1)$ parties. This means that the first $(d-1)$ parties receive the same input $x \in [m]$. For simplicity, their combined inputs are denoted as $\overline{x} = (x,\ldots,x)$.
The output of the $i$-th party (among these $(d-1)$ parties), denoted by $a_i$, corresponds to one of the vertices of the graph $a_i \in [n]$. Collectively, their outputs are represented as $\Vec{a} = (a_0, \ldots, a_{d-2})$. Note that the initial party is labeled as the zeroth party. The last party receives an input $y \in C_x$ and produces a binary output $b \in \{0,1\}$. The winning condition given the inputs $\overline{x},y$ is such that: $(1)$ the outputs of the first $(d-1)$ parties must belong to $C_x$, that is, $a_i\in C_x$ and be distinct, and $(2)$ the output of the last party must be $b=0$ if his input $y$ matches one of the outputs of other parties; otherwise $b=1$. Mathematically, the winning condition is expressed as: for any $x$ and $y\in C_x,$
\begin{equation}\label{eq:optcond}
\sum_{\Vec{a} \in \mathcal{P}(C_x \setminus \{y\})} p(\Vec{a},1|\overline{x},y) + \sum_{\substack{k\in C_x \\ k\neq y}} \sum_{\Vec{a} \in \mathcal{P}(C_x \setminus \{k\})} p(\vec{a},0|\overline{x},y) = 1,
\end{equation}
where $\mathcal{P}(\{\cdot\})$ denotes all possible permutations of the set $\{\cdot\}$ and $\mathcal{P}(C_x \setminus \{k\})$ refers to all possible permutations of the elements of $C_x$ excluding $k$.

The best classical strategies for all the parties can be assumed to be deterministic. In such a strategy, the last party assigns binary values 0 or 1 to all the vertices. Let $v(y)\in \{0,1\}$ be the value assigned to vertex $y$ by the last party. Since the last part does not know the input of other parties, this value assignment must be independent of the input $x$ of the other parties, that is, the context in which $y$ belongs. Without loss of generality, we can assume deterministic output strategies for other parties as well. In each context $C_x,$ there are $d$ vertices, and all other $(d-1)$ parties should produce distinct outputs according to the winning condition $(1)$. This means that for any pair of vertices, say, $y$ and $y'$, at least one of them must be the output of one of these $(d-1)$ parties. Consequently, $v(y)$ and $v(y')$ cannot be 1 simultaneously; otherwise, the winning condition $(2)$ would be violated. Moreover, in a graph that realizes a complete KS set, every pair of connected vertices belongs to at least one context. This requirement forces the condition that no two connected vertices can both be assigned the value 1. Also, to meet the winning conditions, at least one vertex from each context must be assigned the value 1. Consequently, the value assignment must simultaneously satisfy (i) and (ii), which is inherently contradictory for any KS set.

In the quantum strategy, the parties share $\ket{\mathcal{S}_d^{(d)}}$. The first $(d-1)$ parties measure in the basis corresponding to the context $C_x$ for their input $x$, while the last party measures the rank-one projector associated with the vertex $y$. Due to \eqref{prop1}, the condition in \eqref{eq:optcond} is satisfied for all $x,y\in C_x$, ensuring that the winning condition of the game is perfectly achieved. 
\end{proof}

%%%%%%%%%%%%%%%%%%%%%%%%%%%%%%%%%%%%%%%%%%%%%%%%%%%%%%%%%%%%%%%%%%%

\section{Self-testing of Supersinglets}

%%%%%%%%%%%%%%%%%%%%%%%%%%%%%%%%%%%%%%%%%%%%%%%%%%%%%%%%%%%%%%%%%%%

Self-testing \cite{Yao_self,Supic:2020Q} is a method to prove that certain source is actually preparing a specific quantum state using solely from the input-output statistics of a Bell inequality experiment. It is based on the observation that, modulo local isometries, certain input-output statistics can only be produced by a certain state. 
Consider an $d$-party Bell experiment, where we observe the input-output statistics $p(\Vec{a}|\Vec{x})$, with $\Vec{x}=(x_0,x_1,\ldots,x_{d-1})$ representing the measurement settings and $\Vec{a}=(a_0,a_1,\ldots,a_{d-1})$ denoting the measurement outcomes. These outcomes arise from an unknown state $\ket{\psi}$ and unknown local measurements $\{A_{a_i|x_i}\}_{a_i,x_i}$, where $i$ labels the $i$-th party, with the initial party designated as the zeroth party. The local subsystems involved are of arbitrary dimensions. Bell self-testing of a reference state $\ket{\overline{\psi}}$ and local measurements $\{\overline{A}_{a_i|x_i}\}_{a_i,x_i}$ asserts the existence of unitary operators $U_i$ acting on the Hilbert space of the $i$-th party, such that 
\begin{equation} \label{st} \forall i, \ U_iA_{a_i|x_i}U_i^\dagger = \overline{A}_{a_i|x_i} \otimes \id ,
\end{equation}
and
\begin{equation}
\left(\otimes_{i} U_i \right) \ket{\psi} = \ket{\overline{\psi}} \otimes \ket{\text{aux}}. 
\end{equation} 
The auxiliary state $\ket{\text{aux}}$ does not contribute to the observed statistics. From a mathematical standpoint, the observed statistics imply that the unknown quantum state and measurements must have a unique representation. Hence the name ``self-testing''.

%(when the additional orthogonalities $(1,A)$, $(2,B)$, etc are taken into account).. It is worth noting that these additional orthogonality conditions result in a modified orthogonality graph, rendering the set no longer a complete set. 

While all pure bipartite entangled states in arbitrary dimensions can be self-tested within the Bell scenario \cite{Coladangelo:2017NC,bamps_pra2015,sarkar2019self}, extending Bell self-testing to multipartite qudit states remains a challenge. To date, successful self-testing has been limited to specific states constrained by certain dimensions or classifications \cite{mckague2011self,supic_2018,baccari2018scalable,SarkarGHZ,Santos_2023,Panwar2023,Ranendu:PRA2024,balanz:2024XXX}. To the best of our knowledge, supersinglets $\ket{\mathcal{S}_d^{(d)}}$ has yet to be addressed in the context of self-testing. Here, we bridge this gap by introducing a novel technique that leverages local measurements of rigid KS sets to enable self-testing for this class of states.

For that, we first need to focus on a particular subset of the perfect quantum strategies introduced in the previous section: those that use KS sets that are complete and rigid.

A KS set $\{\ket{\overline{\psi_i}}\}_{i=0}^{n-1}$, which belongs to $\mathbbm{C}^d$, with $d \ge 3$ and satisfies the orthogonality and completeness conditions according to an orthogonality graph $\mathcal{G}$ (in which nodes represent vectors and edges indicate which ones are mutually orthogonal), is said to be {\em rigid} if any other set of projectors $\{\Pi_i\}_{i=0}^{n-1}$ that satisfies the same orthogonality and completeness relations dictated by $\mathcal{G}$ and belonging to an arbitrary (but finite) dimensional Hilbert space $\mathbbm{C}^D$, with $D \ge d$, can be related to the reference KS set by a unitary operator $U$ such that, for all $i$, 
\begin{equation} \label{UnRep} 
U \Pi_i U^\dagger = \ket{\overline{\psi_i}}\!\bra{\overline{\psi_i}} \otimes \id, 
\end{equation} 
where $\id$ denotes the identity.
Among the numerous KS sets identified to date, only a few of them are proven to be rigid. Even well-known small KS sets turn out not to be rigid \cite{Xu:2024PRL}. The 18-vector set in Fig.~\ref{fig1} is rigid \cite{Xu:2024PRL}.
The 31-vector KS set in $\mathbbm{C}^3$, proposed by Conway and Kochen \cite{Peres:1993}, has been recently proven to be rigid \cite{Trandafir:2024b}.
 
We present our two main results concerning self-testing. The first establishes a connection between the rigidity of KS sets and the self-testing of measurements performed by local parties in a perfect quantum strategy (PQS). The second provides the self-testing of $d$-party $d$-level supersinglet using the first result along with the fact that for every dimension $d\geqslant 4$, there exists a rigid KS set \cite{Xu:2024PRL}.

Consider a KS set consisting of \( n \) vectors, with its orthogonality structure represented by a graph \(\mathcal{G}\). If the KS set is not complete, it can be extended by adding additional vectors to form a complete KS set, whose orthogonality graph we denote by \( \mathcal{G}_c \), an extension of the graph \(\mathcal{G}\). According to Proposition \ref{prop1}, there exists a PQS based on this extended graph \( \mathcal{G}_c \). In the PQS, we define the uncharacterized measurements corresponding to the first $(d-1)$ parties as \( \{A_{a_i|x_i}\}_{a_i,x_i} \), with $i$ levels the $i$-th party. The uncharacterized measurement of the last party is denoted by the binary-outcome measurement, \( \{I - B_y, B_y\}_y \), corresponding to the outcomes 0 and 1, respectively. Here, \( a_i \) and $y$ belong to the set of vertices of $\mathcal{G}_c$ and \( x_i \) belongs to the set of \( d \)-cliques (or contexts) of $\mathcal{G}_c$. Let the unknown state shared between the parties be $\ket{\psi}$.

We now formally state the results.

\begin{theorem} \label{thm:st-mea}
A KS set $\{\ket{v_y}\}_{y=0}^{n-1}$ with graph $\mathcal{G}$ is rigid according to \eqref{UnRep}, if and only if, the perfect quantum strategy defined by an extended $\mathcal{G}_c$ self-tests the measurements corresponding to $\mathcal{G}$, that is, for each party there exists local unitary $U_i,$ such that for all $a_i \in [n]$ and $i \in [d-1]$,
\begin{equation}\label{Aks-st}
    U_iA_{a_i|x_i}U_i^\dagger = 
    \begin{cases}
\ket{v_{a_i}}\!\bra{v_{a_i}} \otimes \id, \text{ if } a_i \in C_{x_i}, \\
\mathbbm{O}, \quad \text{ otherwise,}
    \end{cases}
\end{equation} 
and
\begin{equation}\label{Bks-st}
    U_{d-1} B_y U_{d-1}^\dagger = \ket{v_y}\!\bra{v_y} \otimes \id, \ \
    \forall y \in [n].
\end{equation}
\end{theorem}

The theorem stated above establishes a one-to-one correspondence between the rigidity of the original KS set (associated with the graph \(\mathcal{G}\)) and the self-testing of the subset of measurement operators \( \{A_{a_i|x_i}\} \) and \( \{B_y\} \), where \( a_i, y \in [n] \), corresponding to the vertices of \( G \).
Note that even if the original KS set is rigid, its extension to a complete KS set may not retain this rigidity. In such cases, the PQS enables the self-testing of the subset of measurements corresponding to the original graph \(\mathcal{G}\). The proof is given in Appendix~\ref{app1}.

\begin{theorem} \label{thm:st-ss}
 There exists a KS set for every Hilbert space ${\cal H}=\mathbbm{C}^d$ of finite $d \ge 3$ such that the corresponding perfect quantum strategy self-tests the $d$-party $d$-level supersinglet.
\end{theorem}

Theorem~\ref{thm:st-ss} establishes that, for every finite $d \geq 3$, there exists a PQS that certifies the unknown state $\ket{\psi}$ to be equivalent to the supersinglet $\ket{\mathcal{S}_d^{(d)}}$, up to local unitaries. The proof is carried out in two steps. The first step relies on two key elements: the existence of rigid KS sets in every dimension \cite{Xu:2024PRL,Trandafir:2024b} and Theorem~\ref{thm:st-mea}. Together, these two results allow for the self-testing of local measurements associated with these rigid KS sets. In the second step, this self-testing of measurements is used to establish the self-testing of the supersinglet in every dimension $d\geq 3$. The detailed proof is presented in Appendix~\ref{app2}.

%%%%%%%%%%%%%%%%%%%%%%%%%%%%%%%%%%%%%%%%%%%%%%%%%%%%%%%%%%%%%%%%%%%

\section{Conclusions and challenges}

%%%%%%%%%%%%%%%%%%%%%%%%%%%%%%%%%%%%%%%%%%%%%%%%%%%%%%%%%%%%%%%%%%%

We intuited that supersinglets were special in quantum theory and quantum information. Here we have made progress in turning this intuition into a proof. We have seen that, indeed, the nonlocality of the supersinglets is special: it is $d$-partite, $d$-dimensional, maximal (in fact, it is ``perfect'', since it corresponds to nonlocal fraction unity \cite{Elitzur:1992PLA}), and, as is the case for Greenberger-Horne-Zeilinger and graph states \cite{mckague2011self}, provides a distinctive signature of the supersinglets. We believe that the quantum-to-classical separation in the perfect quantum strategy is small, in general. For instance, by considering all possible local deterministic strategies, we find that the maximum success probability of the game in the classical setting is 35/36 for the 18-vector set \cite{Cabello:1996PLA} and 59/60 for the 24-vector set \cite{Peres:1991JPA}. Therefore, developing a robust self-testing scheme for supersinglets remains an interesting direction for future work.

The method used to prove it, based on locally measuring rigid KS sets, is interesting in itself. It shows that KS sets have more applications that what is usually appreciated. It is worth investigating whether the same approach can be used to self-test a more general class of multipartite high-dimensional states. For example, it seems clear that, using the same the strategy, rigid KS sets allow us to self-test any $N$-partite high-dimensional state in which, for every bipartition with $N-1$ parties on one partition and one party on the other partition, the $N-1$ parties can predict with certainty the value of all the observables of the KS set corresponding to the other party.

Finally, we also have to see to what extent the entanglement of supersinglets is also special and we have to convince the experimentalists that it is worthwhile to prepare supersinglets. The technology is already available. For example, the $3$-$3$-supersinglet can be prepared with three trapped ions \cite{Wilhelm:2023NC}. For the $4$-$4$-supersinglet, we could use a crosstalk-free processor with eight superconducting qubits \cite{Chen:2022PRL}.
Let us hope that all this and similar efforts will contribute to stimulating the experimental generation of these beautiful states.

%%%%%%%%%%%%%%%%%%%%%%%%%%%%%%%%%%%%%%%%%%%%%%%%%%%%%%%%%%%%%%%%%%%

\subsection*{Acknowledgment}

%%%%%%%%%%%%%%%%%%%%%%%%%%%%%%%%%%%%%%%%%%%%%%%%%%%%%%%%%%%%%%%%%%%

We acknowledge useful conversations with Gabriele Cobucci, Armin Tavakoli, and Stefan Trandafir. This work was supported by the project STARS (Project No.\ STARS/STARS-2/2023-0809) funded by the Government of India, and the EU-funded project \href{10.3030/101070558}{FoQaCiA} and the \href{10.13039/501100011033}{MCINN/AEI} (Project No.\ PID2020-113738GB-I00).

%%%%%%%%%%%%%%%%%%%%%%%%%%%%%%%%%%%%%%%%%%%%%%%%%%%%%%%%%%%%%%%%%%%

\onecolumngrid
\appendix

%%%%%%%%%%%%%%%%%%%%%%%%%%%%%%%%%%%%%%%%%%%%%%%%%%%%%%%%%%%%%%%%%%%

\section{Proof of Theorem \ref{thm:st-mea}}\label{app1}

%%%%%%%%%%%%%%%%%%%%%%%%%%%%%%%%%%%%%%%%%%%%%%%%%%%%%%%%%%%%%%%%%%%

\begin{proof}
The contrapositive statement of the direct implication of this theorem is straightforward to establish. Assume that there exist two complete KS sets associated with the same orthogonality graph $\mathcal{G}$, but these sets are not related by a unitary transformation. Now, consider two PQS realizations corresponding to these two distinct sets of measurements. Since a unitary transformation does not connect the measurements, it follows that the local measurements do not exhibit self-testing properties.

To show the reverse implication of the theorem, it suffices to establish that in any perfect quantum strategy, the last party's measurements $\{\id-B_y, B_y\}$ and the first $(d-1)$ parties' measurements $\{A_{a_i|x_i}\}_{a_i,x_i}$, with $i\in [d-1]$, must be projective and form a KS set according to graph $\mathcal{G}$. If this holds, then these uncharacterized measurements must be unitarily equivalent to the KS set whenever the KS set is \textit{rigid}.

Let $\rho$ be the shared state in the strategy, and $d_i$ be the local dimension of the reduced state for the $i$-th party. Here, the local dimensions $d_i$ can be arbitrary and unknown. Our aim is to characterize the POVMs for every party that acts on the local support of $\rho$.

We define $\rho_{\Vec{a}|\overline{x}} = \tr_{[d-2]}(\rho \bigotimes_{i=0}^{d-2} A_{a_i|x})$, which is the unnormalized reduced state on the $(d-1)$-th party when outcome $\Vec{a}$ is observed by the others for the measurement setting $x$. Here, $\tr_{[d-2]}$ refers to tracing out the subsystems of the first $(d-1)$ parties. The condition $\sum_{\Vec{a}} \bigotimes_{i=0}^{d-2} A_{a_i|x} = \id$ implies that
\begin{equation}\label{eq:sum1}
\sum_{\Vec{a}} \tr_d(\rho_{\Vec{a}|\overline{x}}) = \tr_d \left(\sum_{\Vec{a}} \rho_{\Vec{a}|\overline{x}}\right) = 1.
\end{equation}
The condition for the perfect strategy given by Eq.~\eqref{eq:optcond} translates to
\begin{equation}\label{eq:sum2}
\sum_{\Vec{a} \in \mathcal{P}(C_x \setminus \{y\})} \tr(\rho_{\Vec{a}|\overline{x}} B_y) + \sum_{\substack{k\in C_x \\ k\neq y}} \sum_{\Vec{a} \in \mathcal{P}(C_x \setminus \{k\})} \tr\left[\rho_{\Vec{a}|\overline{x}} (\id-B_y)\right]=1.
\end{equation}
We define $S_{x,y}$ as the subspace spanned by the operator 
\begin{equation}
 \rho_{x,y} = \sum_{\Vec{a} \in \mathcal{P}(C_x \setminus \{y\})} \rho_{\Vec{a}|\overline{x}},
\end{equation} 
and $\tilde{S}_{x,y}$ as the subspace spanned by the operator 
\begin{equation}
 \tilde{\rho}_{x,y} = \sum_{\substack{k\in C_x \\ k\neq y}} \sum_{\Vec{a} \in \mathcal{P}(C_x \setminus \{k\})} \rho_{\Vec{a}|\overline{x}}.
\end{equation} 
Note that $S_{x,y}$ and $\tilde{S}_{x,y}$ depends only on the inputs $x,y$. Consequently, we can re-express Eq.~\eqref{eq:sum1} as 
\begin{equation}\label{eq:sum3}
\tr (\rho_{x,y}) + \tr(\tilde{\rho}_{x,y}) = 1.
\end{equation}
By combining Eqs.~\eqref{eq:sum2} and \eqref{eq:sum3}, we get
\begin{equation}
\label{aeq}
 \tr\left( \rho_{x,y} (\id - B_y) \right)
+ \tr\left( \tilde{\rho}_{x,y} B_y \right) = 0.
\end{equation} 
Since $\mathbbm{O} \preceq B_y \preceq \id$, it follows from Eq.~\eqref{aeq} that the restriction of $B_y$ to the subspace $S_{x,y}$ is the identity operator, and the restriction of $B_y$ to the subspace $\tilde{S}_{x,y}$ is zero. Moreover, because 
$\tilde{S}_{x,y}\oplus S_{x,y}$ spans the entire space of the $(d-1)$-th party's system, we have
\begin{equation}\label{ByP}
B_y = \id_{x,y},
\end{equation}
where $\id_{x,y}$ stands for the identity or projection operator on $S_{x,y}$.

Given the same $x$, for a different input $y'\in C_x$, we similarly find that $B_{y'} = \id_{x,y'}$. Notably, the possible sets of outcomes $\Vec{a}$ that belongs to $\mathcal{P}(C_x\setminus \{y'\})$ and $\mathcal{P}(C_x\setminus \{y\})$ are disjoint, implying that the subspaces $S_{x,y'}$ and $S_{x,y}$ are orthogonal. This leads to the relation
\begin{equation}\label{Byortho}
B_y B_{y'} = \mathbbm{O}, 
\end{equation}
for all $y,y' \in C_x$. Furthermore, Eq.~\eqref{eq:sum3} ensures that $\bigoplus_{y\in C_x} S_{x,y}$ spans the entire space of the last party's subsystem. This gives us the completeness relation
\begin{equation}\label{Bysum}
\sum_{y\in C_x} B_y = \id,
\end{equation}
where $\id$ is the identity operator acting on the entire space of the $(d-1)$-th party's subsystem.
The combined results of Eqs.~\eqref{ByP}, \eqref{Byortho}, and \eqref{Bysum} show that $\{B_y\}$ must form a KS set of projectors satisfying the orthogonality relations according to graph $\mathcal{G}$. Thus, if the KS set is rigid, then \eqref{Bks-st} holds.

A similar analysis applies to the first $(d-1)$ parties. Let us denote the combined reduced states of these parties as $\sigma_y = \tr_{d-1}(\rho B_y), \tilde{\sigma}_y = \tr_{d-1}(\rho (\id-B_y))$, for the measurement setting $y$, such that
\begin{equation} \label{sumtrA}
\tr(\sigma_y) + \tr(\tilde{\sigma}_y) = 1.
\end{equation}
The winning condition in Eq.~\eqref{eq:optcond} implies that
\begin{equation}
 \label{Aaxsum1}
\sum_{\Vec{a} \in \mathcal{P}(C_x \setminus \{y\})} \tr(\sigma_y \bigotimes_i A_{a_i|x}) 
+ \sum_{\substack{k \in C_x \\ k \neq y}} \sum_{\Vec{a} \in \mathcal{P}(C_x \setminus \{k\})} 
 \tr(\tilde{\sigma}_y \bigotimes_i A_{a_i|x}) = 1.
\end{equation}
Before moving forward, let us define the following operators for our convenience:
\begin{align} \label{Axydef}
 A_{x,y} &= \sum_{\Vec{a} \in \mathcal{P}(C_x \setminus \{y\})} \left(\bigotimes_i A_{a_i|x} \right), \\ \label{tildeAxydef}
 \tilde{A}_{x,y} &= \sum_{\substack{k \in C_x \\ k \neq y}} \sum_{\Vec{a} \in \mathcal{P}(C_x \setminus \{k\})} \left( \bigotimes_i A_{a_i|x} \right).
\end{align} 
These operators depend only on inputs $x$ and $y$.
By substituting the right-hand side of Eq.~\eqref{Aaxsum1} with the expression from Eq.~\eqref{sumtrA}, we arrive at the following relation:
\begin{equation} \label{Axysum1}
\tr \left(\sigma_y \left( \id - A_{x,y} \right) \right) + \tr \left[\tilde{\sigma}_y \left( \id - \tilde{A}_{x,y} \right) \right] = 0.
\end{equation}
Given that $\mathbbm{O} \preceq A_{x,y} \preceq \id$ and $\mathbbm{O} \preceq \tilde{A}_{x,y} \preceq \id$, it follows from Eq.~\eqref{tildeAxydef} that
\begin{eqnarray}
 & A_{x,y}|_{\sigma_y} = \id, \label{Asum1} \\
 & A_{x,y}|_{\tilde{\sigma}_y} = \mathbbm{O}, \label{Asum2} \\ 
 & \tilde{A}_{x,y}|_{\tilde{\sigma}_y} = \id, \label{Asum3}
\end{eqnarray}
where $A|_\sigma$ refers to the restriction of operator $A$ to the subspace spanned by $\sigma_y$. Since $\sigma_y$ and $\tilde{\sigma_y}$ cover the entire support of the reduced state of $(d-1)$-party system, Eqs.~\eqref{Asum1} and \eqref{Asum2} imply
\begin{equation} \label{Axyid}
A_{x,y} = \id^{[d-1]}_{x,y},
\end{equation}
with $\id^{[d-1]}_{x,y}$ representing the identity or projection operator acting on the subspace spanned by $\sigma_y,$ and 
\begin{equation}\label{AxyAtxy}
    A_{x,y} + \tilde{A}_{x,y} = \id^{[d-1]}
\end{equation}
where $\id^{[d-1]}$ is the identity operator on the subspace spanned by the reduced state of $(d-1)$ parties' system. 
Replacing $A_{x,y}$ from \eqref{Axydef} in Eq.~\eqref{Axyid}, we find that for every $i$-th party,
\begin{equation} \label{Ayid}
 A_{a_i|x_i} = \id^{(i)}_{x,y},
\end{equation}
where $\id^{(i)}_{x,y}$ denotes the identity or projection operator onto the subspace $i$-th party's subsystem when the reduced state of the combined $(d-1)$ parties is $\sigma_y.$ Furthermore, substituting $A_{x,y}$ and $\tilde{A}_{x,y}$ from \eqref{Axydef} and \eqref{tildeAxydef} in Eq.~\eqref{AxyAtxy}, we get for each $i\in [d-1]$,
\begin{equation} \label{AaiCom}
    \sum_{a_i\in C_{x_i}} A_{a_i|x_i} = \id^{(i)},
\end{equation}
with $\id^{(i)}$ being the identity operator on the system of the $i$-th party. The above also implies $A_{a_i|x_i} = \mathbbm{O}$ if $a_i\notin C_{x_i}$.

Next, due to Eqs.~\eqref{Asum2} and \eqref{Asum3}, we observe that $A_{x,y}$ and $\tilde{A}_{x,y}$ have orthogonal supports. For a different input on the last party, say $y'$, the support of $A_{x,y'}$ must lie within the support of $\tilde{A}_{x,y}$, since $\mathcal{P}(C_x\setminus \{y'\})$ is a subset of $\cup_{k\in C_x, k \neq y} \mathcal{P}(C_x\setminus \{k\})$. This necessitates $A_{x,y}$ and $A_{x,y'}$ are orthogonal for all $y,y' \in C_x$,
\begin{equation} 
 A_{x,y} A_{x,y'} = \mathbbm{O}.
\end{equation}
Substituting the expression of $A_{x,y}$ from Eq.~\eqref{Axydef} into this equation and setting each term equal to zero, we obtain, for every $i\in [d-2]$ and for all pairs $y,y' \in C_{x_i}$,
\begin{equation} \label{Ayortho}
 A_{a_i=y|x_i}A_{a_i=y'|x_i} = \mathbbm{O}.
\end{equation}
 
Hence, Eqs.~\eqref{Ayid}, \eqref{Ayortho}, and \eqref{AaiCom}, together imply that $\{A_{a_i|x_i}\}$ represents a realization of the KS set.
Therefore, if the KS set is \textit{rigid}, then any two sets of local measurements by each party are linked by unitary, admitting self-testing given by \eqref{Aks-st}.
\end{proof}

%%%%%%%%%%%%%%%%%%%%%%%%%%%%%%%%%%%%%%%%%%%%%%%%%%%%%%%%%%%%%%%%%%%

\section{Proof of Theorem \ref{thm:st-ss}}\label{app2}

%%%%%%%%%%%%%%%%%%%%%%%%%%%%%%%%%%%%%%%%%%%%%%%%%%%%%%%%%%%%%%%%%%%

\begin{proof}
From Theorem \ref{thm:st-mea}, we know that, if a KS set $\{\ket{v_i}\}_{i=0}^{n-1}$ in $\mathbbm{C}^d$ is rigid, then the respective PQS implies self-testing of the local measurements given by Eqs.~\eqref{Aks-st} and \eqref{Bks-st}, along with the fact that the Hilbert space of $i$-th party can be decomposed as $\mathcal{H}_i=\mathbbm{C}^d\otimes \mathcal{H}'_i$. 
By substituting the measurements from Eqs.~\eqref{Aks-st} and \eqref{Bks-st} and the transformed unknown state $(\otimes_i U_i)\ket{\psi}$, we find that the probabilities take the form $p(\vec{a},b|\vec{x},y) = \tr\left[ \left(\bigotimes_{i} \overline{A}_{a_i|x_i}\right) \overline{\rho} \right]$, where $\{\overline{A}_{a_i|x_i}\}_{a_i,x_i}$ are the reference measurements corresponding to the KS set of vectors, and 
\begin{equation}
\overline{\rho} = \tr_{\otimes_i\mathcal{H}'_i} \left[(\otimes_iU_i)\ket{\psi}\!\bra{\psi} (\otimes_iU_i)^\dagger \right]
\end{equation}
is the reduced state obtained by tracing out the auxiliary subsystems on which the trivial measurements act [via Eqs.~\eqref{Aks-st} and \eqref{Bks-st}]. This shows that the probabilities in the PQS depend entirely on the reduced state $\overline{\rho}$, as expected from the definition of self-testing. Let the spectral decomposition of the reduced state be $\overline{\rho} = \sum_j c_j \ket{\overline{\psi_j}}\!\bra{\overline{\psi_j}}$, where each $\ket{\overline{\psi_j}} \in (\mathbbm{C}^{d})^{\otimes^n}$ is a $d$-partite quantum state with $d$-dimensional local subsystems. It can be easily checked that as PQS, defined by \eqref{eq:optcond}, is observed from $\overline{\rho}$, it must also be achieved for each $\ket{\overline{\psi_j}}$. Consider any of these states and denote it by $\ket{\overline{\psi}}$. We will show that this state must have a unique form of the supersinglet. Consequently, $\overline{\rho}$ must itself be a pure state of the form $\ket{\overline{\psi}}\!\bra{\overline{\psi}}$ and, furthermore, the unknown state $\otimes_iU_i\ket{\psi}= \ket{\overline{\psi}}\otimes \ket{\text{aux}}$, where the auxiliary state $ \ket{\text{aux}}$ belongs to $\otimes_i \mathcal{H}'_i$.

The general form of $\ket{\overline{\psi}}$ is given by 
\begin{equation} \label{gen-psi}
\ket{\overline{\psi}} = \sum_{i=0}^{d-1}\sum_{s_i=0}^{d-1} \alpha_{s_0s_1\ldots s_{d-1}} \ket{s_0s_1\ldots s_{d-1}},
\end{equation}
where, $\{\ket{s_i}\}$, with $s_i \in [d]$, denotes the canonical basis of the $i$-th party.
For any $x$ corresponding to a basis, the fact that it is a perfect quantum strategy imposes the condition
\begin{equation} \label{pqs-s}
 p(\Vec{a},1|\overline{x},y) = 0 \quad \forall y\in C_x \ \forall \Vec{a} \notin \mathcal{P}(C_x \setminus \{y\}),
\end{equation}
which, when replaced with the quantum expression, leads to
\begin{equation} \label{pqs-v}
\left( \bigotimes_{i=0}^{d-1} |v_{a_i}\rangle\!\langle v_{a_i}| \right) \ket{\overline{\psi}} = 0 \quad \forall (a_0,\ldots,a_{d-1}) \notin \mathcal{P}(C_x).
\end{equation}
Without loss of generality, we can choose one context or basis, say $C_0 \equiv \{0,\ldots,d-1\}$, in the KS set to be the canonical basis, $\{\ket{v_j}\}_{j=0}^{d-1} \equiv \{\ket{t}\}_{t=0}^{d-1}$. Taking that basis, Eq.~\eqref{pqs-s} becomes 
\begin{equation} \label{gen-psi-rel}
 \left( \bigotimes^{d-1}_{i=0} |t_i\rangle\!\langle t_i| \right) \ket{\overline{\psi}} = 0
 \quad \forall (t_0,t_1,\ldots, t_{d-1}) \notin \mathcal{P}([d]),
\end{equation}
where $\{\ket{t_i}\}$ denotes the canonical basis for $i$-th party. By substituting the general of $\ket{\overline{\psi}}$ from \eqref{gen-psi} in the above relation \eqref{gen-psi-rel} and equating the coefficients of each basis vector to zero, we find that
$\alpha_{s_0s_1\ldots s_{d-1}} = 0$ whenever $(s_0,s_1,\ldots, s_{d-1}) \notin \mathcal{P}([d])$. This simplifies the form of $\ket{\overline{\psi}}$ to 
\begin{equation} \label{simpsi}
 \ket{\overline{\psi}} = \sum_{\substack{s_0,s_1, \ldots, s_{d-1} \\
 (s_0,s_1,\ldots, s_{d-1}) \in \mathcal{P}([d])}} \alpha_{s_0s_1\ldots s_{d-1}} \ket{s_0s_1\ldots s_{d-1}}.
\end{equation}

The above analysis holds for any perfect quantum strategy that is based on a complete rigid KS set. From here onward, we will focus on a specific KS set for each $d$ that is rigid and demonstrate that the shared state in Eq.~\eqref{simpsi} is necessarily $d$-party $d$-level supersinglet. Recall, however, that there exist rigid KS sets in any $d \ge 3$ \cite{Xu:2024PRL}.

For $d=3$, we consider the 31-vector KS set proposed by Conway and Kochen \cite{Peres:1993}. The explicit form of this set, which includes the canonical basis, is available in Table~IV of \cite{Trandafir:2024XXX}. This set has been shown to be rigid \cite{Trandafir:2024b}, making Theorem~\ref{thm:st-mea} applicable for the corresponding perfect quantum strategy. Applying Eqs.~\eqref{pqs-s} and \eqref{simpsi} to this case, we obtain 
\begin{equation} \label{pqs-ck31}
\left( \bigotimes_{i=0}^2 |v_{a_i}\rangle\!\langle v_{a_i}| \right) \ket{\overline{\psi}} = 0 \quad \forall (a_0,a_1,a_2) \notin \mathcal{P}(C_x),
\end{equation}
where
\begin{equation} \label{psi-ck31}
 \ket{\overline{\psi}} = \sum_{\substack{i,j,k \in [3] \\ (i,j,k) \in \mathcal{P}([3])}}\alpha_{ijk} \ket{ijk} .
\end{equation}
This KS set contains two bases $\{\ket{v_0}\equiv \ket{0},\ket{v_3},\ket{v_4}\}$ and $\{\ket{v_1}\equiv \ket{1},\ket{v_5},\ket{v_6}\}$, where the vectors are given as follows:
\begin{equation}
\ket{v_3} = \frac{1}{\sqrt{2}}\begin{bmatrix}
 0 \\ 1 \\ -1 
\end{bmatrix},\quad
\ket{v_4} = \frac{1}{\sqrt{2}} \begin{bmatrix}
 0 \\ 1 \\ 1 
\end{bmatrix}, \quad
\ket{v_5} = \frac{1}{\sqrt{2}} \begin{bmatrix}
 1 \\ 0 \\ -1 
\end{bmatrix}, \quad
\ket{v_6} = \frac{1}{\sqrt{2}} \begin{bmatrix}
 1 \\ 0 \\ 1 
\end{bmatrix}.
\end{equation}
Consequently, Eq.~\eqref{pqs-ck31} holds for the two respective contexts $C_1 = \{0,3,4\}$ and $C_2=\{1,5,6\}$. After substituting one triple of outcomes, say $(0,3,3)$, which does not belong to $\mathcal{P}(C_1)$, in Eq.~\eqref{pqs-ck31} with the shared state \eqref{psi-ck31}, yields
\begin{equation}
 \left( |0\rangle\!\langle 0| \otimes |v_3\rangle\!\langle v_3| \otimes |v_3\rangle\!\langle v_3| \right) \sum_{(i,j,k) \in \mathcal{P}([3])} \alpha_{ijk} \ket{ijk} = 0.
\end{equation}
This reduces to 
\begin{equation}
 -\frac12 \left( \alpha_{012} + \alpha_{021} \right) \ket{0}\!\ket{v_3}\!\ket{v_3} = 0,
\end{equation}
which directly implies $\alpha_{012} + \alpha_{021} = 0$. 
Following a similar process, by systematically considering all possible outcomes that are not permutations within the two contexts $C_1$ and $C_2$, and substituting them into Eq.~\eqref{pqs-ck31} for the state \eqref{psi-ck31}, we obtain five independent relations among the coefficients $\{\alpha_{ijk}\}$:
\begin{align}
\alpha_{012} + \alpha_{021} &= 0, \\
\alpha_{102} + \alpha_{201} &= 0, \\
\alpha_{120} + \alpha_{210} &= 0, \\ 
\alpha_{102} + \alpha_{120} &= 0, \\ 
\alpha_{012} + \alpha_{210} &= 0.
\end{align}
These equations are simplified to
\begin{equation}
 \alpha_{012} = - \alpha_{021} = \alpha_{120} = - \alpha_{102} = \alpha_{201} = -\alpha_{210}.
\end{equation} 
Further, the normalization condition on the state ensures that the state must be the three-party three-level supersinglet.

Next, consider the Peres-24 set in $\mathbbm{C}^4$, whose explicit form is given in Table \ref{tab:peres-24}. This set has been proven to be rigid \cite{Xu:2024PRL}, so Theorem \ref{thm:st-mea} applies to the corresponding perfect quantum strategy. Equations~\eqref{pqs-s} and \eqref{simpsi} in this case demand 
\begin{equation} \label{pqs-p24}
\left( \bigotimes_{i=0}^3 |v_{a_i}\rangle\!\langle v_{a_i}| \right) \ket{\overline{\psi}} = 0 ,\quad \forall (a_0,a_1,a_2,a_3) \notin \mathcal{P}(C_x),
\end{equation}
where
\begin{equation} \label{psi-p24}
 \ket{\overline{\psi}} = \sum_{\substack{i,j,k,l \in [4] \\ (i,j,k,l) \in \mathcal{P}([4])}}\alpha_{ijkl} \ket{ijkl} .
\end{equation}
We focus on the two bases $\{\ket{v_4},\ket{v_5},\ket{v_6},\ket{v_7}\}$ and $\{\ket{v_8},\ket{v_{9}},\ket{v_{10}},\ket{v_{11}}\}$, whose explicit forms are provided in Table~\ref{tab:peres-24}. 
Taking similar approach as before, we consider all possible outcomes that are not permutations of the respective contexts $\{4,5,6,7\}$ and $\{8,9,10,11\}$, and substitute them 
in Eq.~\eqref{pqs-p24}, together with the state \eqref{psi-p24}. This results in a set of 23 linearly independent equations involving 24 variables $\alpha_{ijkl}$, where $(i,j,k,l) \in \mathcal{P}([4])$ as follows:
\begin{equation} \label{lin-eqs}
\begin{bmatrix}
0 & 0 & 0 & 0 & 0 & 0 & 0 & 0 & 0 & 0 & 0 & 0 & 0 & 0 & 0 & 0 & 1 & 1 & 0 & 0 & 0 & 0 & 1 & 1 \\
0 & 0 & 0 & 0 & 0 & 0 & 0 & 0 & 0 & 0 & 0 & 0 & 0 & 0 & 0 & 0 & -1 & 1 & 0 & 0 & 0 & 0 & -1 & 1 \\
0 & 0 & 0 & 0 & 0 & 0 & 0 & 0 & 0 & 0 & 0 & 0 & 0 & 0 & 0 & 0 & -1 & -1 & 0 & 0 & 0 & 0 & 1 & 1 \\
0 & 0 & 0 & 0 & 0 & 0 & 0 & 0 & 0 & 0 & 0 & 0 & 0 & 1 & 0 & 1 & 0 & 0 & 0 & 1 & 0 & 1 & 0 & 0 \\
0 & 0 & 0 & 0 & 0 & 0 & 0 & 0 & 0 & 0 & 0 & 0 & 0 & -1 & 0 & 1 & 0 & 0 & 0 & -1 & 0 & 1 & 0 & 0 \\
0 & 0 & 0 & 0 & 0 & 0 & 0 & 0 & 0 & 0 & 0 & 0 & 0 & -1 & 0 & -1 & 0 & 0 & 0 & 1 & 0 & 1 & 0 & 0 \\
0 & 0 & 0 & 0 & 0 & 0 & 0 & 0 & 0 & 0 & 0 & 0 & 1 & 0 & 1 & 0 & 0 & 0 & 1 & 0 & 1 & 0 & 0 & 0 \\
0 & 0 & 0 & 0 & 0 & 0 & 0 & 0 & 0 & 0 & 0 & 0 & -1 & 0 & -1 & 0 & 0 & 0 & 1 & 0 & 1 & 0 & 0 & 0 \\
0 & 0 & 0 & 0 & 0 & 0 & 0 & 0 & 0 & 0 & 0 & 0 & -1 & 0 & 1 & 0 & 0 & 0 & -1 & 0 & 1 & 0 & 0 & 0 \\
0 & 0 & 0 & 1 & 0 & 1 & 0 & 0 & 0 & 1 & 0 & 1 & 0 & 0 & 0 & 0 & 0 & 0 & 0 & 0 & 0 & 0 & 0 & 0 \\
0 & 0 & 0 & -1 & 0 & -1 & 0 & 0 & 0 & 1 & 0 & 1 & 0 & 0 & 0 & 0 & 0 & 0 & 0 & 0 & 0 & 0 & 0 & 0 \\
0 & 0 & 0 & -1 & 0 & 1 & 0 & 0 & 0 & -1 & 0 & 1 & 0 & 0 & 0 & 0 & 0 & 0 & 0 & 0 & 0 & 0 & 0 & 0 \\
0 & 0 & 1 & 0 & 1 & 0 & 0 & 0 & 1 & 0 & 1 & 0 & 0 & 0 & 0 & 0 & 0 & 0 & 0 & 0 & 0 & 0 & 0 & 0 \\
0 & 0 & -1 & 0 & 1 & 0 & 0 & 0 & -1 & 0 & 1 & 0 & 0 & 0 & 0 & 0 & 0 & 0 & 0 & 0 & 0 & 0 & 0 & 0 \\
0 & 0 & -1 & 0 & -1 & 0 & 0 & 0 & 1 & 0 & 1 & 0 & 0 & 0 & 0 & 0 & 0 & 0 & 0 & 0 & 0 & 0 & 0 & 0 \\
1 & 1 & 0 & 0 & 0 & 0 & 1 & 1 & 0 & 0 & 0 & 0 & 0 & 0 & 0 & 0 & 0 & 0 & 0 & 0 & 0 & 0 & 0 & 0 \\
-1 & 1 & 0 & 0 & 0 & 0 & -1 & 1 & 0 & 0 & 0 & 0 & 0 & 0 & 0 & 0 & 0 & 0 & 0 & 0 & 0 & 0 & 0 & 0 \\
-1 & -1 & 0 & 0 & 0 & 0 & 1 & 1 & 0 & 0 & 0 & 0 & 0 & 0 & 0 & 0 & 0 & 0 & 0 & 0 & 0 & 0 & 0 & 0 \\
0 & 0 & 0 & 0 & 0 & 0 & 0 & 0 & 0 & 0 & 1 & 1 & 0 & 0 & 0 & 0 & 0 & 0 & 0 & 0 & 1 & 1 & 0 & 0 \\
0 & 0 & 0 & 0 & 0 & 0 & 0 & 0 & 0 & 0 & -1 & 1 & 0 & 0 & 0 & 0 & 0 & 0 & 0 & 0 & -1 & 1 & 0 & 0 \\
0 & 0 & 0 & 0 & 0 & 0 & 0 & 0 & 0 & 0 & -1 & -1 & 0 & 0 & 0 & 0 & 0 & 0 & 0 & 0 & 1 & 1 & 0 & 0 \\
0 & 0 & 0 & 0 & 0 & 0 & 0 & 1 & 0 & 1 & 0 & 0 & 0 & 0 & 0 & 0 & 0 & 0 & 1 & 0 & 0 & 0 & 0 & 1 \\
0 & 0 & 0 & 0 & 0 & 0 & 0 & -1 & 0 & 1 & 0 & 0 & 0 & 0 & 0 & 0 & 0 & 0 & -1 & 0 & 0 & 0 & 0 & 1 
\end{bmatrix}
\begin{bmatrix}
 \alpha_{0123} \\
\alpha_{0132} \\
 \alpha_{0213} \\
 \alpha_{0 2 3 1}\\
 \alpha_{0 3 1 2}\\
 \alpha_{0 3 2 1}\\
 \alpha_{1 0 2 3}\\
 \alpha_{1 0 3 2}\\
 \alpha_{1 2 0 3}\\
 \alpha_{1 2 3 0}\\
 \alpha_{1 3 0 2}\\
 \alpha_{1 3 2 0}\\
 \alpha_{2 0 1 3}\\
 \alpha_{2 0 3 1}\\
 \alpha_{2 1 0 3}\\
 \alpha_{2 1 3 0}\\
 \alpha_{2 3 0 1}\\
 \alpha_{2 3 1 0}\\
 \alpha_{3 0 1 2}\\
 \alpha_{3 0 2 1}\\
 \alpha_{3 1 0 2}\\
 \alpha_{3 1 2 0}\\
 \alpha_{3 2 0 1}\\
 \alpha_{3 2 1 0}\\
\end{bmatrix} = 
\begin{bmatrix}
 0 \\
 0 \\
 0 \\
 0 \\
 0 \\
 0 \\
 0 \\
 0 \\
 0 \\
 0 \\
 0 \\
 0 \\
 0 \\
 0 \\
 0 \\
 0 \\
 0 \\
 0 \\
 0 \\
 0 \\
 0 \\
 0 \\
 0 \\
 0 \\
\end{bmatrix}.
\end{equation}
Solving these gives the following relations: 
\begin{eqnarray} \label{1234}
 && \alpha_{0123} = - \alpha_{0132} = - \alpha_{0213} = \alpha_{0231} = \alpha_{0312} = - \alpha_{0321} \nonumber \\
 &=& - \alpha_{2134} = \alpha_{2143} = \alpha_{2314} = - \alpha_{2341} = - \alpha_{2413} = \alpha_{2431} \nonumber \\ 
 &=& \alpha_{2013} = - \alpha_{2031} = - \alpha_{2103} = \alpha_{2130} = \alpha_{2301} = - \alpha_{2310} \nonumber \\ 
 &=& - \alpha_{3012} = \alpha_{3021} = \alpha_{3102} = - \alpha_{3120} = - \alpha_{3201} = \alpha_{3210}.
\end{eqnarray}
Using the above relations with the normalization condition, we conclude that $\ket{\overline{\psi}}$ is the supersinglet of $d=4$. 

To extend this to $d$-party, $d$-level supersinglet with $d\geqslant 4$, we employ the $d$-dimensional KS set introduced in \cite{Xu:2024PRL}. These sets have been shown to be rigid for all $d$ \cite{Xu:2024PRL}, and thus, Theorem~\ref{thm:st-mea} is applicable. This $d$-dimensional KS set is constructed by merging $(d-3)$ Peres-24 sets of vectors, each in four-dimensional subspaces. To define if explicitly, let $\mathcal{S}_k$ denote the subspace spanned by the canonical basis vectors $\{\ket{k},\ket{k+1},\ket{k+2},\ket{k+3}\}$, where $k \in [d-3]$. Let $\ket{v^k_i}$ represent the $d$-dimensional vector such that the $i$-th vector from Table~\ref{tab:peres-24} appears in the subspace of $\mathcal{S}_k$, with all other elements in the respective vector set to zero. The KS set is then defined as,
\begin{equation}
 \bigcup_{k=0}^{d-4} \left\{\ket{v^k_i} \right\}_{i=0}^{23} ,
\end{equation}
where $\ket{v_i}$ are given in Table \ref{tab:peres-24}.

Since this KS set includes the canonical basis, we know that the state $\ket{\overline{\psi}}$ must be of the form \eqref{simpsi}. It can be noted that the following two bases are present in this KS set,
\begin{equation} \label{b1k}
\{\ket{v^{k}_4}, \ket{v^{k}_5}, \ket{v^{k}_6}, \ket{v^{k}_7}\} \cup \{\ket{t} \}_{t\neq k,k+1,k+2,k+3}
\end{equation} and 
\begin{equation} \label{b2k}
\{\ket{v^{k}_8}, \ket{v^{k}_{9}}, \ket{v^{k}_{10}},\ket{v^{k}_{11}}\} \cup \{\ket{t} \}_{t\neq k,k+1,k+2,k+3}
\end{equation} 
for every $k \in [d-3]$, where $\{\ket{t}\}$ denotes the canonical basis. 

Let us first focus on the 24 unknown coefficients in Eq.~\eqref{simpsi} that are of the form $\alpha_{P(s_0,s_1,s_2,s_3,s_4,s_5,\ldots, s_{d-1})}$, in which we fix the values of $s_4,s_5,\ldots, s_{d-1}$, and $(s_0,s_1,s_2,s_3) \in \mathcal{P}([4])$. Hereafter we will use the notation $P(\cdot )$ to denote any permutation of the respective set. Taking two bases given by Eqs.~\eqref{b1k} and \eqref{b2k}, with $k=0$, and substituting them in Eq.~\eqref{pqs-v} along with the state \eqref{simpsi}, we get the same set of equations listed in \eqref{lin-eqs}. Consequently, these 24 coefficients must satisfy the relations \eqref{1234}. Therefore, given any values of $s_4,\ldots, s_{d-1}$, and any permutation $P(\cdot)$ we have 
\begin{equation} \label{al1234p}
\alpha_{P(0,1,2,3,s_4,\ldots ,s_{d-1})} = \varepsilon_{s_0 s_1 s_2 s_3 s_4 \ldots s_{d-1}} \alpha_{P(s_0,s_1,s_2,s_3,s_4,\ldots ,s_{d-1})},
\end{equation}
where $\varepsilon_{s_0 s_1 s_2 s_3 s_4 \ldots s_{d-1}}$ is the Levi-Civita symbol.
 
The next step is to determine the coefficients of the form $\alpha_{s_0,s_1,s_2,s_3,s_4,\ldots ,s_{d-1}}$, in which $(s_0,s_1,s_2,s_3,s_4) \in \mathcal{P}[\{0,1,2,3,s_4\}]$ and other $s_i$'s are fixed. By taking different $P(\cdot)$ in Eq.~\eqref{al1234p}, we obtain the following relations:
\begin{align}
 \alpha_{0,1,2,3,s_4,s_5,\ldots ,s_{d-1}} & = \varepsilon_{s_0 s_1 s_2 s_3 s_4 s_5 \ldots s_{d-1}} \alpha_{s_0,s_1,s_2,s_3,s_4,s_5,\ldots ,s_{d-1}}, \label{p4-1} \\
 \alpha_{0,1,2,s_4,3,s_5,\ldots ,s_{d-1}} & = \varepsilon_{s_0 s_1 s_2 s_4 s_3 s_5\ldots s_{d-1}} \alpha_{s_0,s_1,s_2,s_4,s_3,s_5,\ldots ,s_{d-1}}, \label{p4-2} \\
 \alpha_{0,1,s_4,3,2,s_5,\ldots ,s_{d-1}} & = \varepsilon_{s_0 s_1 s_4 s_3 s_2 s_5\ldots s_{d-1}} \alpha_{s_0,s_1,s_4,s_3,s_2,s_5,\ldots ,s_{d-1}}, \label{p4-3} \\
 \alpha_{0,s_4,2,3,1,s_5,\ldots ,s_{d-1}} & = \varepsilon_{s_0 s_4 s_2 s_3 s_1 s_5\ldots s_{d-1}} \alpha_{s_0,s_4,s_2,s_3,s_1,s_5,\ldots ,s_{d-1}}, \label{p4-4} \\
 \alpha_{s_4,1,2,3,0,s_5,\ldots ,s_{d-1}} & = \varepsilon_{s_4 s_1 s_2 s_3 s_0 s_5\ldots s_{d-1}} \alpha_{s_4,s_1,s_2,s_3,s_0,s_5,\ldots ,s_{d-1}}, \label{p4-5}
\end{align}
where $(s_0,s_1,s_2,s_3) \in \mathcal{P}([4])$. For every permutation, the above five equations involve distinct sets of coefficients. So we have to find relations between coefficients that appear in different sets. For that, let us again take two bases \eqref{b1k} and \eqref{b2k}, with $k=1$, and substitute them in Eq.~\eqref{pqs-v} with the state \eqref{simpsi}. This will yield 
\begin{equation}\label{al234s5p}
\alpha_{P(0,1,2,3,s_4,s_5,\ldots,s_{d-1})} = \varepsilon_{0 s_1 s_2 s_3 s_4 s_5\ldots s_{d-1}} \alpha_{P(0,s_1,s_2,s_3,s_4,s_5,\ldots,s_{d-1})},
\end{equation}
for $(s_1,s_2,s_3,s_4)\in \mathcal{P}(\{1,2,3,s_4\})$. Taking the trivial permutation in Eq.~\eqref{al234s5p}, we have
%ERE
\begin{align}
 \alpha_{0,1,2,3,s_4,s_5,\ldots,s_{d-1}} &= - \alpha_{0,1,2,s_4,3,s_5,\ldots,s_{d-1}}, \label{s51}\\
 \alpha_{0,1,2,3,s_4,s_5,\ldots,s_{d-1}} &= - \alpha_{0,1,s_4,3,2,s_5,\ldots,s_{d-1}}, \label{s52}\\
 \alpha_{0,1,2,3,s_4,s_5,\ldots,s_{d-1}} &= - \alpha_{0,s_4,2,3,1,s_5,\ldots,s_{d-1}}. \label{s53} 
\end{align}
While, taking the permutation $P(\cdot)$ in Eq.~\eqref{al234s5p} as the transposition between $0$ and $s_4$, we obtain
\begin{equation}
 \alpha_{s_4,1,2,3,0,s_5,\ldots,s_{d-1}} = - \alpha_{1,s_4,2,3,0,s_5,\ldots,s_{d-1}},
\end{equation}
 which using Eq.~\eqref{p4-2} implies
 \begin{equation}
 \label{s54} \alpha_{0,1,2,3,s_4,s_5,\ldots,s_{d-1}} = - \alpha_{s_4,1,2,3,0,s_5,\ldots,s_{d-1}}.
 \end{equation}
Subsequently, Eqs.~\eqref{al1234p}, \eqref{p4-1}-\eqref{p4-5}, and \eqref{s51}-\eqref{s54} imply
\begin{equation}\label{al1234s5}
\alpha_{0,1,2,3,s_4,s_5,\ldots,s_{d-1}} = \varepsilon_{s_0 s_1 s_2 s_3 s_4 s_5\ldots s_{d-1}}\alpha_{s_0,s_1,s_2,s_3,s_4,s_5,\ldots,s_{d-1}}
\end{equation}
for $(s_0,s_1,s_2,s_3,s_4)\in \mathcal{P}(\{0,1,2,3,s_4\})$ and for any fixed values of $s_5\ldots s_{d-1}$. The same computation, from Eq.~\eqref{al1234p} to Eq.~\eqref{al1234s5}, can be done to prove 
\begin{equation}\label{als12345}
\alpha_{s_0,s_1,s_2,s_3,s_4,s_5,\ldots,s_{d-1}} = \varepsilon_{s_i s_j s_k s_l s_4 s_5\ldots s_{d-1}} \alpha_{s_i,s_j,s_k,s_l,s_4,s_5,\ldots,s_{d-1}}
\end{equation}
for $(s_i,s_j,s_k,s_l)\in \mathcal{P}(\{s_0,s_1,s_2,s_3\})$ for any $s_0,s_1,s_2,s_3 \in [d]$ and any fixed values of $s_4 s_5\ldots s_{d-1}$. For instance, to prove this relation \eqref{als12345} for $s_0=0,s_1=1$ and any $s_2=6,s_3=8$, we first consider Eq.~\eqref{al1234s5} with $s_4=6$
and then consider \eqref{als12345} again with $s_4=8$ with the first four elements to be $0,1,6,3$. 

To get the relation between all possible permutations of $s_0\ldots s_5$, we can fix different values of $s_0$ and do the same analysis by substituting the bases \eqref{b1k} and \eqref{b2k} by taking $k=1$ and $k=2$. This process can be executed further recursively up to $s_{d-1}$ since it is possible to generate all possible permutations $\mathcal{P}([d])$ from the canonical order by taking the transposition of two consecutive elements. This leads to our desired relation, 
\begin{equation}
 \alpha_{0,1,\ldots,{d-1}} = \varepsilon_{s_0 s_1\ldots s_{d-1}} \alpha_{s_0,s_1,\ldots,s_{d-1}}
\end{equation}
for all $(s_0,\ldots, s_{d-1}) \in \mathcal{P}([d])$. Finally, using the normalization condition, we conclude that the state must be the $d$-party $d$-level supersinglet.
\end{proof}

%%%%%%%%%%%%%%%%%%%%%%%%%%%%%%%%%%%%%%%%%%%%%%%%%%%%%%%%%%%%%%%%%%%
% Table I
%%%%%%%%%%%%%%%%%%%%%%%%%%%%%%%%%%%%%%%%%%%%%%%%%%%%%%%%%%%%%%%%%%%

\begin{center}
\begin{table}[ht!]
\begin{tabular}{cccccccccccccccccccccccc}
\hline 
$\ket{v_0}$ & $\ket{v_1}$ & $\ket{v_2}$ & $\ket{v_3}$ & $\ket{v_4}$ & $\ket{v_5}$ & $\ket{v_6}$ & $\ket{v_7}$ & $\ket{v_8}$ & $\ket{v_{9}}$ & $\ket{v_{10}}$ & $\ket{v_{11}}$ & $\ket{v_{12}}$ & $\ket{v_{13}}$ & $\ket{v_{14}}$ & $\ket{v_{15}}$ & $\ket{v_{16}}$ & $\ket{v_{17}}$ & $\ket{v_{18}}$ & $\ket{v_{19}}$ & $\ket{v_{20}}$ & $\ket{v_{21}}$ & $\ket{v_{22}}$ & $\ket{v_{23}}$ \\
\hline 
&&&&&&&&&&&&&&&&&&&&&&& \\
$\begin{bmatrix}
 1 \\ 0 \\ 0 \\ 0 
\end{bmatrix}$ & 
$\begin{bmatrix}
 0 \\ 1 \\ 0 \\ 0 
\end{bmatrix}$ & 
$\begin{bmatrix}
 0 \\ 0 \\ 1 \\ 0 
\end{bmatrix}$ & 
$\begin{bmatrix}
 0 \\ 0 \\ 0 \\ 1 
\end{bmatrix}$ & 
$\begin{bmatrix}
 1 \\ 1 \\ 0 \\ 0 
\end{bmatrix}$ & 
$\begin{bmatrix}
 1 \\ \bar{1} \\ 0 \\ 0 
\end{bmatrix}$ & 
$\begin{bmatrix}
 0 \\ 0 \\ 1 \\ 1 
\end{bmatrix}$ & 
$\begin{bmatrix}
 0 \\ 0 \\ 1 \\ \bar{1} 
\end{bmatrix}$ & 
$\begin{bmatrix}
 1 \\ 0 \\ 1 \\ 0 
\end{bmatrix}$ & 
$\begin{bmatrix}
 1 \\ 0 \\ \bar{1} \\ 0 
\end{bmatrix}$ & 
$\begin{bmatrix}
 0 \\ 1 \\ 0 \\ 1 
\end{bmatrix}$ & 
$\begin{bmatrix}
 0 \\ 1 \\ 0 \\ \bar{1} 
\end{bmatrix}$ & 
$\begin{bmatrix}
 1 \\ 1 \\ 1 \\ 1 
\end{bmatrix}$ & 
$\begin{bmatrix}
 1 \\ 1 \\ \bar{1} \\ \bar{1} 
\end{bmatrix}$ & 
$\begin{bmatrix}
 1 \\ \bar{1} \\ 1 \\ \bar{1} 
\end{bmatrix}$ & 
$\begin{bmatrix}
 1 \\ \bar{1} \\ \bar{1} \\ 1 
\end{bmatrix}$ & 
$\begin{bmatrix}
 1 \\ 1 \\ 1 \\ \bar{1} 
\end{bmatrix}$ & 
$\begin{bmatrix}
 1 \\ 1 \\ \bar{1} \\ 1 
\end{bmatrix}$ & 
$\begin{bmatrix}
 1 \\ \bar{1} \\ 1 \\ 1 
\end{bmatrix}$ & 
$\begin{bmatrix}
 \bar{1} \\ 1 \\ 1 \\ 1 
\end{bmatrix}$ & 
$\begin{bmatrix}
 1 \\ 0 \\ 0 \\ 1 
\end{bmatrix}$ & 
$\begin{bmatrix}
 1 \\ 0 \\ 0 \\ \bar{1}
\end{bmatrix}$ & 
$\begin{bmatrix}
 0 \\ 1 \\ 1 \\ 0 
\end{bmatrix}$ & 
$\begin{bmatrix}
 0 \\ 1 \\ \bar{1} \\ 0 
\end{bmatrix}$ \\
\end{tabular}
\caption{Peres-24 set of vectors (unnormalized), where $\bar{1}=-1$.}
\label{tab:peres-24}
\end{table}
\end{center}

%%%%%%%%%%%%%%%%%%%%%%%%%%%%%%%%%%%%%%%%%%%%%%%%%%%%%%%%%

%\bibliographystyle{apsrev4-2} %Remove to allow the longbibliography option to work
%\bibliography{common2}

%

%%%%%%%%%%%%%%%%%%%%%%%%%%%%%%%%%%%%%%%%%%%%%%%%%%%%%%%%%%%%%%%%%%%

\end{document}